\documentclass[runningheads]{llncs}
\usepackage[T1]{fontenc}
\usepackage{graphicx}
\usepackage[frozencache]{minted}
\usepackage{lipsum}
\usepackage{color}
\usepackage{minted}
\usepackage{qcircuit}
\usepackage{amsmath}
\usepackage{amssymb}
\usepackage{braket}
\usepackage[hidelinks]{hyperref}

\let\oldtexttt\texttt
\renewcommand{\texttt}[1]{\mbox{\oldtexttt{#1}}}

\setminted[python3]{frame=lines,breaklines,fontsize=\footnotesize}

\begin{document}

\title{Optimizing Quantum Compilation via\\High-Level Quantum Instructions}

\author{Evandro C. R. Rosa \and
    Jerusa Marchi \and\\
    Eduardo I. Duzzioni \and
    Rafael de Santiago}
\authorrunning{E. C. R. Rosa et al.}
\institute{Universidade Federal de Santa Catarina, Florianópolis, Brazil
    \email{evandro.crr@posgrad.ufsc.br} \email{\{jerusa.marchi,eduardo.duzzioni,r.santiago\}@ufsc.br}}

\maketitle        
\begin{abstract}
    Current quantum programming is dominated by low-level, circuit-centric approaches that limit the potential for compiler optimization. This work presents how a high-level programming construct provides compilers with the semantic information needed for advanced optimizations. We introduce a novel optimization that leverages a quantum-specific instruction to automatically substitute quantum gates with more efficient, approximate decompositions, a process that is transparent to the programmer and significantly reduces quantum resource requirements. Furthermore, we show how this instruction guarantees the correct uncomputation of auxiliary qubits, enabling safe, dynamic quantum memory management. We illustrate these concepts by implementing a V-chain decomposition of the multi-controlled NOT gate, showing that our high-level approach not only simplifies the code but also enables the compiler to generate a circuit with up to a 50\% reduction in CNOT gates. Our results suggest that high-level abstractions are crucial for unlocking a new class of powerful compiler optimizations, paving the way for more efficient quantum computation.

    \keywords{Quantum Computing \and Quantum Programming \and Compiler Optimization \and Ket}
\end{abstract}

\section{Introduction}

Quantum computers have the potential to solve problems that are intractable for classical computers. Although this has long been known~\cite{shorPolynomialTimeAlgorithmsPrime1997,groverFastQuantumMechanical1996}, the first demonstration of a quantum advantage was presented only in 2019~\cite{aruteQuantumSupremacyUsing2019}. Still, this demonstration solved a problem with no practical application. Quantum computing is an emerging technology where new developments in both software and hardware are required to enable its practical use.

On the software side, quantum programming is still based on low-level constructs such as qubits and quantum gates. Quantum programs are typically constructed by explicitly defining a quantum circuit, as is done with platforms like Qiskit~\cite{Javadi-Abhari2024}, or by manipulating qubits directly, as in Q\#~\cite{svoreEnablingScalableQuantum2018} and Ket~\cite{darosaKetQuantumProgramming2022}. The latter approach, while rooted in low-level gate applications, is beginning to offer higher-level instructions.

Despite the predominantly low-level nature of current quantum programming, we argue that higher-level coding paradigms are already emerging. The generated quantum program cannot be executed directly by a quantum computer and must undergo a compilation process in a classical computer. Furthermore, as addressed in this paper, certain instructions can reduce the lines of code and enable optimizations.

The quantum compilation process can be divided into three main steps. First, multi-qubit gates are decomposed into sequences of one- and two-qubit gates~\cite{Rosa2025}. While programming languages allow for the convenient use of multi-qubit gates, the underlying quantum hardware is often limited to performing only single- and two-qubit gates. Second, logical qubits are mapped to physical qubits~\cite{zhuDynamicLookAheadHeuristic2020,liTacklingQubitMapping2019}. This step must account for limitations in connectivity; while logical qubits are assumed to be fully connected, physical qubits can typically only interact with their immediate neighbors. The third step is to translate these one- and two-qubit gates into the native gate set of the target quantum computer. Although a quantum computer can perform universal computation, it implements a limited set of native gates. Once the program is decomposed into sequences of native gates that respect the hardware's connectivity, calibration data is used to generate the pulse sequences necessary to physically manipulate the qubits~\cite{Lussi2025}. This final step is typically performed by the quantum computer's controller immediately before execution~\cite{stefanazziQICKQuantumInstrumentation2022}.

This work builds upon the Ket quantum programming platform by presenting how a quantum-specific instruction, which implements an operation of the form $U^\dagger V U$, enables optimizations in the early stages of compilation, particularly during quantum gate decomposition. In addition, we show how this instruction facilitates the safe allocation and deallocation of auxiliary qubits, thereby allowing dynamic quantum memory management. The main contributions of this paper are:
\begin{itemize}
    \item Enabling approximate quantum gate decomposition, leading to more efficient circuits in a way that is transparent to the programmer.
    \item Ensuring that auxiliary qubits are disentangled and returned to the zero state before deallocation.
\end{itemize}

While dynamic quantum memory management does not immediately enable compilation optimization, it opens an avenue for future improvements, where the compiler has more freedom in mapping auxiliary qubits to physical qubits. Additionally, auxiliary qubit allocation allows for the implementation of functions and gates with simpler interfaces. The caller only needs to manage the primary qubits involved in the operation, while the auxiliary qubits are managed transparently by the compiler.

This paper is structured as follows. Section~\ref{sec:qp} presents the Ket quantum programming platform, with a focus on the high-level \texttt{with around} instruction and how multi-qubit gates arise naturally in quantum programming. Section~\ref{sec:opt} details the compiler optimizations enabled by this instruction, and Section~\ref{sec:aux} describes its role in the safe management of auxiliary qubits. Section~\ref{sec:ex} provides an example of the use of the \texttt{with around} instruction, focusing on the implementation of decomposition algorithms to demonstrate the performance impact of the proposed optimization. Finally, Section~\ref{sec:conclusion} presents our final remarks and outlines future work.

For the remainder of this paper, we assume the reader is familiar with the mathematical formalism of quantum computing and quantum circuit diagrams. For a general introduction to quantum computing, we refer the reader to the textbook Nielsen and Chuang~\cite{nielsenQuantumComputationQuantum2010}.

\section{Quantum Programming}\label{sec:qp}

In this section, we provide an introduction to quantum programming with Ket. The objective is not an exhaustive presentation, but rather to introduce the concepts and instructions relevant to the proposed optimizations. For a more in-depth introduction to Ket, we refer the reader to the project's official website\footnote{\url{https://quantumket.org}} and to some introductory papers presenting the platform~\cite{Rosa2025a,darosaKetQuantumProgramming2022}.

Quantum programming in Ket operates by directly manipulating the state of qubits, in contrast to circuit-centric platforms like Qiskit~\cite{Javadi-Abhari2024}. In Ket, qubits are first-class objects, and quantum gates are treated as functions that take qubits as input. The platform includes eight built-in single-qubit gates: the Pauli gates (\texttt{X}, \texttt{Y}, and \texttt{Z}), the rotation gates (\texttt{RX}, \texttt{RY}, and \texttt{RZ}), the phase gate (\texttt{P}), and the Hadamard gate (\texttt{H}). These are sufficient to prepare a single qubit in any arbitrary state.

Universal quantum computation requires multi-qubit gates. Although Ket does not provide built-in multi-qubit gates, it achieves universality by allowing any operation to be controlled. This is a core design principle of the platform. Any function that calls quantum gates (and does not allocate or measure qubits) is itself considered a quantum gate. This allows for the creation of complex, reusable operations that can also be controlled. For example, Figure~\ref{fig:ket:cnot} shows two equivalent implementations of a CNOT gate, created by applying a control qubit to a built-in \texttt{X} gate. Ket provides two primary ways to apply control: the \texttt{with control} context manager and the \texttt{ctrl()} function.

\begin{figure}[htbp]
    \centering
    \begin{minipage}{.31\linewidth}
        \begin{minted}[fontsize=\normalsize]{python3}
def my_cnot(c, t):
    with control(c):
        X(t)
  \end{minted}
    \end{minipage}
    \hfil
    \begin{minipage}{.28\linewidth}
        \begin{minted}[fontsize=\normalsize]{python3}
def my_cnot(c, t):
    ctrl(c, X)(t)

  \end{minted}
    \end{minipage}
    \caption{Equivalent CNOT gate implementations in Ket: \textsc{Left} via a \texttt{with control} block, \textsc{Right} via the \texttt{ctrl()} function.}
    \label{fig:ket:cnot}
\end{figure}

Quantum gates are unitary transformations, meaning they are reversible ($UU^\dagger  = U^\dagger U = I$). Ket provides the \texttt{adj()} function to obtain the adjoint (inverse) of any gate. As shown in Figure~\ref{fig:ket:rxx}, the function \texttt{rxx\_xplct} illustrates this capability by demonstrating the use of \texttt{adj(U)}.

%

\begin{figure}[htbp]
    \centering
    \begin{minipage}[t]{.49\linewidth}
        \begin{minted}{python3}
def rxx_xplct(angle: float, qubits):
    U = cat(kron(H, H), CNOT)
    U(qubits[0], qubits[1])
    RZ(angle, qubits[1])
    adj(U)(qubits[0], qubits[1])
  \end{minted}
    \end{minipage}
    \hfil
    \begin{minipage}[t]{.41\linewidth}
        \begin{minted}{python3}
def rxx(angle: float, qubits):
    U = cat(kron(H, H), CNOT)
    with around(U, *qubits):
        RZ(angle, qubits[1])

\end{minted}
    \end{minipage}
    \centering
    \includegraphics[width=.6\linewidth]{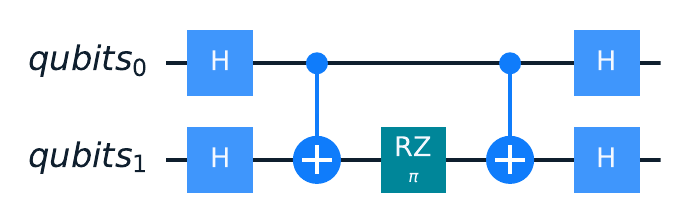}

    \caption{Ket implementation of the $R_{XX}$ gate. \textsc{Top Left}: The \texttt{rxx\_xplct} function implements the gate by explicitly calling the inverse of the \texttt{U} gate. \textsc{Top Right}: The \texttt{rxx} function uses the \texttt{with around} instruction, which automatically applies the inverse of the gate \texttt{U} at the end of the block. \textsc{Bottom}: Resulting quantum circuit with \texttt{angle} equal to $\pi$. Both codes generate the same quantum circuit.}
    \label{fig:ket:rxx}
\end{figure}

Many quantum algorithms use the pattern $U^\dagger V U$, where a transformation $V$ is applied within a unitary $U$. To simplify this common structure, Ket provides the \texttt{with around} instruction. This instruction automatically applies a gate $U$ at the beginning of a code block and its inverse $U^\dagger$ at the end. As shown in Figure~\ref{fig:ket:rxx}, the function \texttt{rxx\_xplct} implements the $R_{XX}$ gate by explicitly applying the operator $U$ and its adjoint. In contrast, the \texttt{rxx} function achieves the same result using \texttt{with around}. While both implementations are functionally equivalent, the \texttt{with around} instruction provides the compiler with semantic information that can be leveraged for optimizations, as we will discuss in the next section.

%
%
%
%
%
%
%
%

\section{Compiling Optimization}\label{sec:opt}

The use of the \texttt{with around} instruction to optimize quantum computations was first proposed by Rosa \textit{et al.}~\cite{rosaOptimizingGateDecomposition2025}. Their work focused on reducing the overhead of applying control to a gate that is implemented with this instruction. In this work, we propose a novel optimization that leverages the \texttt{with around} instruction to safely apply approximate gate decompositions. This approach requires fewer quantum resources while guaranteeing that the final result is correct.

This section is organized as follows: Section~\ref{subsec:opt:ctrl} reviews the existing optimization for controlled operations, and Section~\ref{subsec:opt:decomp} introduces our new method for using approximate decompositions.

\subsection{Controlled Gate Reduction}\label{subsec:opt:ctrl}

A controlled gate $C^nU$, where $U$ is a unitary operation and $n$ is the number of control qubits, can be defined by its action on the control and target qubits:
\begin{equation}\label{eq:cu}
    C^nU = \sum_{k=0}^{2^n-2}\ket{k}\!\!\bra{k}\otimes I + \ket{2^n{-1}}\!\!\bra{2^n{-1}}\otimes U
\end{equation}
This means the unitary $U$ is applied to the target qubits \textit{if and only if} all $n$ control qubits are in the state $\ket{1}$. The state of the control qubits are not altered by the operation.

Given a quantum gate $U$ that is composed of a sequence of gates,
\begin{equation}
    U = U_k \cdots U_1 U_0 \equiv
    \Qcircuit @C=1em @R=.7em {
    & {/}\qw&\gate{U} & \push{\rule{.1em}{0em}=\rule{.1em}{0em}}\qw & \gate{U_0} & \gate{U_1} & \push{\rule{.1em}{0em}\cdots\rule{.1em}{0em}}\qw &\gate{U_k}& \qw
    },
\end{equation}
its controlled version can be decomposed by distributing the control over each gate in the sequence. This relationship is shown below:
\begin{equation}
    C^nU = C^nU_k\cdots  C^nU_1 C^nU_0\equiv
    \begin{array}{c}
        \Qcircuit @C=1em @R=.7em {
         & {/}^{n}\qw & \ctrl{1} & \qw \\
         & {/} \qw    & \gate{U} & \qw
        }
    \end{array}=
    \begin{array}{c}
        \Qcircuit @C=1em @R=.7em {
         & {/}^{n}\qw & \ctrl{1}   & \ctrl{1}   & \push{\rule{.1em}{0em}\cdots\rule{.1em}{0em}}\qw & \ctrl{1}   & \qw \\
         & {/} \qw    & \gate{U_0} & \gate{U_1} & \push{\rule{.1em}{0em}\cdots\rule{.1em}{0em}}\qw & \gate{U_k} & \qw
        }
    \end{array}
\end{equation}
In Ket, this decomposition is performed recursively until the operation consists of a sequence of controlled built-in gates.

Quantum gates constructed using the \texttt{with around} instruction take the form $U = A^\dagger B A$. A naive decomposition of its controlled version would be
\begin{equation}
    C^nU = C^n(A^\dagger)\cdot C^nB\cdot C^nA.
\end{equation}
However, the controls on the $A$ and $A^\dagger$ gates can be eliminated. This optimization is formalized in the following theorem.

\begin{theorem}\label{theorem:rmc}
    Let $U$ be a unitary operation of the form $U = A^\dagger B A$, where $A$ and $B$ are also unitary. The n-controlled version of $U$, denoted $C^nU$, can be simplified as:
    \begin{equation}
        C^n(A^\dagger B A) = A^\dagger (C^nB) A
    \end{equation}
    Visually, this equivalence is:
    \begin{equation}
        \begin{array}{c}
            \Qcircuit @C=1em @R=.7em {
             & {/}^n\qw & \ctrl{1} & \qw \\
             & {/}\qw   & \gate{U} & \qw
            }
        \end{array}=\begin{array}{c}
            \Qcircuit @C=1em @R=.7em {
             & {/}^n\qw & \qw      & \ctrl{1} & \qw              & \qw \\
             & {/}\qw   & \gate{A} & \gate{B} & \gate{A^\dagger} & \qw
            }
        \end{array}
    \end{equation}
\end{theorem}

\begin{proof}
    The proof considers two cases based on the state of the control qubits. First, if all control qubits are in the state $\ket{1}$, the controlled-$B$ operation is applied, which results in the operation $A^\dagger B A$ on the target qubit(s), the desired transformation. Second, if the control qubits are not all in the state $\ket{1}$, the controlled-$B$ acts as the identity. The resulting operation is $A^\dagger I A = A^\dagger A$. Since $A$ is unitary, this product simplifies to the identity $I$, correctly leaving the target qubits unchanged.
    \qed
\end{proof}

\subsection{Approximated Decomposition}\label{subsec:opt:decomp}

An approximate decomposition of a unitary operation is one that differs from the exact operation only by a local phase, \textit{i.e.}, $\overline{U} = D U$ where $D$ is a diagonal unitary. Such decompositions are often useful for reducing the cost of multi-qubit gates. A well-known example is the approximate decomposition of the Toffoli gate~\cite{maslovAdvantagesUsingRelativephase2016} shown Eq.~\eqref{eq:toffoli_apx}. It uses only four CNOT gates instead of the standard six, but differs from the exact Toffoli gate by a phase factor on one of the input states ($U\ket{010} = -\ket{010}$). The standard decomposition of the Toffoli gate is shown in Eq.~\eqref{eq:toffoli}.

\begin{equation}\label{eq:toffoli_apx}
    \begin{array}{c}
        \Qcircuit @C=.5em @R=1.5em {
         & \qw                          & \ctrl{2} & \qw                          & \qw      & \qw                         & \ctrl{2} & \qw                         & \qw \\
         & \qw                          & \qw      & \qw                          & \ctrl{1} & \qw                         & \qw      & \qw                         & \qw \\
         & \gate{{R_Y}(\frac{-\pi}{4})} & \targ    & \gate{{R_Y}(\frac{-\pi}{4})} & \targ    & \gate{{R_Y}(\frac{\pi}{4})} & \targ    & \gate{{R_Y}(\frac{\pi}{4})} & \qw
        }
    \end{array}
    = \begin{bmatrix}
        1 & 0 & 0  & 0 & 0 & 0 & 0 & 0 \\
        0 & 1 & 0  & 0 & 0 & 0 & 0 & 0 \\
        0 & 0 & -1 & 0 & 0 & 0 & 0 & 0 \\
        0 & 0 & 0  & 1 & 0 & 0 & 0 & 0 \\
        0 & 0 & 0  & 0 & 1 & 0 & 0 & 0 \\
        0 & 0 & 0  & 0 & 0 & 1 & 0 & 0 \\
        0 & 0 & 0  & 0 & 0 & 0 & 0 & 1 \\
        0 & 0 & 0  & 0 & 0 & 0 & 1 & 0 \\
    \end{bmatrix}
\end{equation}

\begin{equation}\label{eq:toffoli}
    \begin{array}{c}
        \Qcircuit @C=.5em @R=1.5em {
         & \qw      & \qw      & \qw              & \ctrl{2} & \qw      & \qw      & \qw              & \ctrl{2} & \qw      & \ctrl{1} & \gate{T}         & \ctrl{1} & \qw \\
         & \qw      & \ctrl{1} & \qw              & \qw      & \qw      & \ctrl{1} & \qw              & \qw      & \gate{T} & \targ    & \gate{T^\dagger} & \targ    & \qw \\
         & \gate{H} & \targ    & \gate{T^\dagger} & \targ    & \gate{T} & \targ    & \gate{T^\dagger} & \targ    & \gate{T} & \gate{H} & \qw              & \qw      & \qw \\
        }
    \end{array}
    =
    \begin{bmatrix}
        1 & 0 & 0 & 0 & 0 & 0 & 0 & 0 \\
        0 & 1 & 0 & 0 & 0 & 0 & 0 & 0 \\
        0 & 0 & 1 & 0 & 0 & 0 & 0 & 0 \\
        0 & 0 & 0 & 1 & 0 & 0 & 0 & 0 \\
        0 & 0 & 0 & 0 & 1 & 0 & 0 & 0 \\
        0 & 0 & 0 & 0 & 0 & 1 & 0 & 0 \\
        0 & 0 & 0 & 0 & 0 & 0 & 0 & 1 \\
        0 & 0 & 0 & 0 & 0 & 0 & 1 & 0 \\
    \end{bmatrix}
\end{equation}

Because these decompositions are not exactly equivalent to the original operation, they can only be used in specific cases where the erroneous local phases cancel out. Our proposal is that the compiler can leverage the structure of the \texttt{with around} instruction to automatically identify circuits where it is safe to substitute a gate with its more efficient approximate version.

For the purpose of this paper, we define two classes of gates. A \emph{permutation gate} is any unitary that has only one non-zero element in each row and column, such as the Pauli gates. A \emph{diagonal gate} is a unitary with non-zero elements only on its main diagonal, such as the Pauli-\texttt{Z} and \texttt{RZ} gates. It is important to note that a controlled-permutation gate is also a permutation gate, and a controlled-diagonal gate is still a diagonal gate.

\begin{theorem}\label{theorem:apx}
    Given a unitary operation of the form $P^\dagger D P$, where $P$ is a permutation gate, any approximate decomposition of $P$ that differs only by a local phase may be used without altering the final unitary, provided that the operator $D$ preserves the computational basis of the qubits transformed by $P$.
\end{theorem}

Preserving the computation basis means that if a basis state $\ket{k}$ enters $D$, the output must be of the form $e^{i\theta_k}\ket{k}$ (globally, the state may change in other qubits, but not in the one $P$ permutes). For example, $D$ itself could be a controlled gate where the qubits in $P$ acts upon are the control qubits.

The intuition is that the phase error introduced by the approximate $P$ is exactly canceled by the conjugate phase error from its adjoint, $P^\dagger$. This cancellation is guaranteed by the structure of the circuit. The general form of a circuit that satisfies the conditions of Theorem~\ref{theorem:apx} is shown below, where $P$ is a permutation gate, $D$ is a diagonal gate, and $C^nU$ is an arbitrary controlled operation.
\begin{equation}\label{eq:apx}
    \Qcircuit @C=1em @R=.7em {
    & {/}\qw & \gate{P}  & \gate{D}  & \ctrl{1} & \gate{P^\dagger} & \qw   \\
    & {/}\qw &  \qw    & \qw     & \gate{U} &  \qw       & \qw
    }
\end{equation}

\begin{proof}
    Our proof consists of calculating the unitary for the circuit in Eq.~\eqref{eq:apx} and showing that the result is identical when $P$ is replaced with an approximate version, $\overline{P}$.

    First, we demonstrate that a diagonal gate $D$ acting on a control qubits commutes with a controlled-unitary $C^nU$ acting on a target qubits. Let $D$ and $C^nU$ be defined as:
    \begin{equation}
        D = \sum_{k=0}^{2^n{-1}}e^{i\theta_k}\ket{k}\!\!\bra{k}, \quad C^nU = \sum_{j=0}^{2^n-2}\ket{j}\!\!\bra{j}\otimes I + \ket{2^n{-1}}\!\!\bra{2^n{-1}}\otimes U\label{eq:dcu}
    \end{equation}
    The product of these operators is:
    \begin{equation}
        (D\otimes I) \cdot C^nU = C^nU \cdot (D\otimes I) = \sum_{k=0}^{2^n-2}e^{i\theta_k}\ket{k}\!\!\bra{k}\otimes I + e^{i\theta_{2^n{-1}}}\ket{2^n{-1}}\!\!\bra{2^n{-1}}\otimes U
    \end{equation}
    Thus, the order of $D$ and $C^nU$ does not matter. Now, let the permutation gate $P$ be defined by a permutation function $p(k)$ and a set of local phases $\{\phi_k\}$:
    \begin{equation}
        P = \sum_{k=0}^{2^n{-1}}e^{i\phi_k}\ket{p(k)}\!\!\bra{k}, \quad P^\dagger = \sum_{k=0}^{2^n{-1}}e^{-i\phi_k}\ket{k}\!\!\bra{p(k)}\label{eq:p}
    \end{equation}
    Let us analyze the state of the circuit from Eq.~\eqref{eq:apx} just before the final $P^\dagger$ gate is applied:
    \begin{equation}
        \begin{aligned}
             & C^nU\cdot (D \cdot P)\otimes I \\ &\quad= \sum_{k=0}^{2^n-2}e^{i(\theta_{p(k)} +\phi_k)}\ket{p(k)}\!\!\bra{k}\otimes I \\&\quad\quad+ e^{i(\theta_{p(2^n-1)}+\phi_{2^n{-1}})}\ket{p(2^n-1)}\!\!\bra{2^n{-1}}\otimes U
        \end{aligned}
    \end{equation}
    Finally, applying the $P^\dagger$ gate will cancel out the phases $\theta_k$ introduced by the $P$ gate:
    \begin{equation}\label{eq:final_opt}
        \begin{aligned}
             & (P^\dagger\otimes I)\cdot C^nU\cdot ((D \cdot P)\otimes I)                                                                       \\
             & \quad = \sum_{k=0}^{2^n-2}e^{i(\theta_p(k) +{\color{red}\phi_k-\phi_k})}\ket{p(k)}\!\!\bra{p(k)}\otimes I                        \\
             & \quad\quad  + e^{i(\theta_{p(2^n{-1})}+{\color{red}\phi_{2^n{-1}}-\phi_{2^n{-1}}})}\ket{p(2^n{-1})}\!\!\bra{p(2^n{-1})}\otimes U
        \end{aligned}
    \end{equation}
    Let an approximate version of $P$, denoted $\overline{P}$, differ by a diagonal phase gate $D_p$, such that $\overline{P} = D_p P$:
    \begin{equation}\label{eq:dp}
        \overline{P} = \left(\sum_{k=0}^{2^n{-1}}e^{i\gamma_k}\ket{k}\!\!\bra{k}\right) \left(\sum_{k=0}^{2^n{-1}}e^{i\phi_k}\ket{p(k)}\!\!\bra{k}\right) = \sum_{k=0}^{2^n{-1}}e^{i(\gamma_{p(k)}+\phi_k)}\ket{p(k)}\!\!\bra{k}
    \end{equation}
    Note that the $P$ and $\overline{P}$ gates only differ by the local phases $\gamma_k$, which are cancelled in the same manner as the $\phi_k$ terms. Therefore, replacing $P$ with $\overline{P}$ in the circuit of Eq.~\ref{eq:apx} results in the same final unitary.
    \qed
\end{proof}

Another way to prove this is to note that the diagonal gate $D_p$ defining the approximation commutes with the other gates. Therefore, we can rearrange the expression so that $D_p$ and $D_p^\dagger$ cancel out:
\begin{equation}
    \begin{aligned}
         & (\overline{P}^\dagger \otimes I) \cdot (D \otimes I) \cdot C^nU \cdot (\overline{P} \otimes I)                    \\
         & \quad= (P^\dagger D_p^\dagger \otimes I) \cdot (D \otimes I) \cdot C^nU \cdot (D_p P \otimes I)                   \\
         & \quad= (P^\dagger \otimes I) \cdot (D \otimes I) \cdot C^nU \cdot (P \otimes I) \cdot (D_p^\dagger D_p \otimes I) \\
         & \quad= (P^\dagger \otimes I) \cdot (D \otimes I) \cdot C^nU \cdot (P \otimes I)
    \end{aligned}
\end{equation}

\section{Auxiliar Qubit Allocation}\label{sec:aux}

Auxiliary qubits assist in the application of a given operation. They are not strictly necessary for the implementation of a unitary, but they can facilitate its decomposition and make the execution more efficient. One example of auxiliary qubit usage is in the decomposition of multi-controlled gates. For instance, for the $n$-controlled NOT gate, the best-known algorithm without auxiliary qubits uses $O(n^2)$ CNOT gates~\cite{dasilvaLineardepthQuantumCircuits2022}. However, given enough auxiliary qubits, it is possible to decompose it in $O(n)$ CNOTs~\cite{Rosa2025}.

An auxiliary qubit must be returned to its initial state after the operation is complete. This process, often called uncomputation~\cite{bichselSilqHighlevelQuantum2020}, ensures that the qubit is not entangled with the rest of the system and cannot generate unwanted interference in subsequent operations. There are two kinds of auxiliary qubits, clean and dirty, depending on their state before the operation. \emph{Clean} auxiliary qubits are guaranteed to be in the state $\ket{0}$, while \emph{dirty} ones can be in an unknown state. While clean auxiliary qubits may be less available, they usually result in more efficient implementations~\cite{Rosa2025}.

Ket automatically manages auxiliary qubits for internal quantum gate decompositions, ensuring they are returned to their original state. In this work, however, we propose using the \texttt{with around} instruction to safely expose auxiliary qubits to the programmer. The primary goal is to ensure that an auxiliary qubit is correctly returned to its initial state after an operation. We propose a pessimistic strategy to validate that a clean auxiliary qubit remains in its original state, and thus can be safely freed to be used in other operations without side effects from entanglement.

\begin{theorem}\label{theorem:alloc}
    Given the circuit below, where the auxiliary qubits are initialized in the state $\ket{\alpha} = \ket{0\dots 0}$, the qubits' state returns to its initial state upon completion of the full sequence of operations.
    \begin{equation}\label{eq:alloc}
        \Qcircuit @C=1.2em @R=1em {
        \lstick{\ket{\psi} }     & \qw /^{n}  & \ctrl{2} & \qw & \qw     & \ctrl{2}     & \qw \\
        \lstick{\ket{\varphi} }     & \qw /^{m}  & \qw    & \qw & \gate{U}  & \qw        & \qw \\
        \lstick{\ket{\alpha} }  &  \qw /^{a}  & \gate{P} & \gate{D} & \ctrl{-1} & \gate{P^\dagger} & \qw
        }
    \end{equation}
    This holds given that $P$ is a permutation gate, $D$ is a diagonal gate, and $U$ is an arbitrary unitary gate. The number of qubits in each subsystem is $n, a > 0$ and $m \geq 0$.
\end{theorem}

The circuit from Eq.~\eqref{eq:alloc} is equivalent to the following unitary operation, where we use subscripts to indicate the qubits on which the gates and controls act:
\begin{equation}
    (C^n_\psi P^\dagger_\alpha \otimes I_{\varphi}) \cdot (I_{\psi} \otimes C^a_\alpha U_\varphi) \cdot (I_{\psi} \otimes I_{\varphi}\otimes D_\alpha)  \cdot (C^n_\psi P_\alpha \otimes I_{\varphi})
\end{equation}
Defining the unitaries $P$ and $D$ as in Eq.~\eqref{eq:p} and \eqref{eq:dcu}, we see this circuit structure matches the $A^\dagger B A$ pattern of the \texttt{with around} instruction, where we have:
\begin{equation}
    A = (C^n_\psi P_\alpha \otimes I_{\varphi}), \quad B = (I_{\psi} \otimes C^a_\alpha U_\varphi) \cdot (I_{\psi} \otimes I_{\varphi}\otimes D_\alpha)\label{eq:alloc:ab}
\end{equation}

\begin{proof}
    Our proof consists of showing that the circuit in Eq.~\eqref{eq:alloc} only applies a local phase to the auxiliary qubits $\ket{\alpha}$. Since this qubits are initialized in the state $\ket{0\dots0}$, a local phase keep the subsystems separable:
    \begin{equation}
        \ket{\nu}\otimes e^{i\theta} \ket{0\dots0} = e^{i\theta} \ket{\nu}\otimes \ket{0\dots0},
    \end{equation}
    where $\ket{\nu}$ is the arbitrary state of the other qubits.

    We start by computing the unitary $B$ from Eq.~\eqref{eq:alloc:ab}
    \begin{equation}
        B =  I_\psi \otimes \left[ I_\varphi \otimes \sum_{k=0}^{2^a{-2}}e^{i\theta_k}\ket{k}\!\!\bra{k}_{\alpha} \;+\;  U_\varphi \otimes e^{i\theta_{2^a{-1}}}\ket{2^a{-1}}\!\!\bra{2^a{-1}}_{\alpha}\right],
    \end{equation}
    and then the final unitary $A^\dagger B A$
    \begin{equation}
        \begin{aligned}
            A^\dagger B A = & \sum_{j=0}^{2^n{-2}} \ket{j}\!\!\bra{j}_\psi &  & \otimes &  & I_\varphi &  & \otimes &  & I_\alpha                                                           \\
                            & + \ket{2^n{-1}}\!\!\bra{2^n{-1}}_\psi        &  & \otimes &  & I_\varphi &  & \otimes &  & \sum_{k=0}^{2^a{-2}}e^{i\theta_k}\ket{p(k)}\!\!\bra{p(k)}_{\alpha} \\
                            & + \ket{2^n{-1}}\!\!\bra{2^n{-1}}_\psi        &  & \otimes &  & U_\varphi &  & \otimes &  & e^{i\theta_{2^a{-1}}}\ket{p(2^a{-1})}\!\!\bra{p(2^a{-1})}_{\alpha}
        \end{aligned}
    \end{equation}

    Since the state of the auxiliary qubits $\ket{\alpha}$ are initialized to the $\ket{0\dots0}$ state, we can analyze the circuit's action by applying the projector $\ket{0}\!\!\bra{0}_\alpha$ to the final unitary without changing the outcome. We now analyze the result in two cases, depending on the permutation function $p$ of the unitary $P$.
    \begin{description}
        \item[Case 1] When $p(2^a{-1}) = 0$,  the final unitary is:
            \begin{equation}
                \begin{aligned}
                    A^\dagger B A\ket{0}\!\!\bra{0}_\alpha = & \begin{aligned}[t]
                                                                    & \sum_{j=0}^{2^n{-2}} \ket{j}\!\!\bra{j}_\psi &  & \otimes &  & I_\varphi &  & \otimes &  & \ket{0}\!\!\bra{0}_{\alpha}                      \\
                                                                    & + \ket{2^n{-1}}\!\!\bra{2^n{-1}}_\psi        &  & \otimes &  & U_\varphi &  & \otimes &  & e^{i\theta_{2^a{-1}}}\ket{0}\!\!\bra{0}_{\alpha} \\
                                                               \end{aligned}                                                                                              \\
                    =                                        & \left[ \sum_{j=0}^{2^n{-2}} \ket{j}\!\!\bra{j}_\psi  \otimes  I_\varphi                                       +  \ket{2^n{-1}}\!\!\bra{2^n{-1}}_\psi \otimes  e^{i\theta_{2^a{-1}}} U_\varphi \right]  \otimes  \ket{0}\!\!\bra{0}_{\alpha}
                \end{aligned}
            \end{equation}
            In this case, a local phase is applied to the auxiliary qubits, which remains unentangled.
        \item[Case 2] When $p(2^a{-1}) \not = 0$,  the final unitary is:
            \begin{equation}
                \begin{aligned}
                    A^\dagger B A\ket{0}\!\!\bra{0}_\alpha = & \begin{aligned}[t]
                                                                    & \sum_{j=0}^{2^n{-2}} \ket{j}\!\!\bra{j}_\psi &  & \otimes &  & I_\varphi &  & \otimes &  & \ket{0}\!\!\bra{0}_{\alpha}                        \\
                                                                    & + \ket{2^n{-1}}\!\!\bra{2^n{-1}}_\psi        &  & \otimes &  & I_\varphi &  & \otimes &  & e^{i\theta_{p^{-1}(0)}}\ket{0}\!\!\bra{0}_{\alpha}
                                                               \end{aligned}                                                  \\
                    =                                        & \left[\sum_{j=0}^{2^n{-2}} \ket{j}\!\!\bra{j}_\psi  \otimes  I_\varphi  + \ket{2^n{-1}}\!\!\bra{2^n{-1}}_\psi   \otimes  e^{i\theta_{p^{-1}(0)}}I_\varphi\right]  \otimes  \ket{0}\!\!\bra{0}_{\alpha}
                \end{aligned}
            \end{equation}
            Where $p^{-1}$ is the inverse function $p$. Again, the auxiliary qubits are only multiplied by a phase and remains separable from the other qubits.
    \end{description}
    \qed
\end{proof}

This circuit structure can be enforced by the Ket compiler. Only diagonal gates are permitted to act directly on an auxiliary qubit qubits. However, when using the \texttt{with around} instruction ($A^\dagger B A$), the unitary $A$ is permitted to be a permutation gate, under the condition that the qubits involved in $A$ are not modified by the inner unitary $B$.

\section{Example: Decomposition Algorithm}\label{sec:ex}

In this section, we present an implementation of the V-chain decomposition algorithm for the multi-controlled NOT gate~\cite{barencoElementaryGatesQuantum1995} to demonstrate the practical benefits of Ket's high-level instruction. Figure~\ref{code:vchain} illustrates two possible implementations: one that leverages the \texttt{with around} instruction and another that uses an explicit approach. The objective is to show how a high-level abstraction can both simplify the programming effort and enable compiler optimizations that improve performance. While this is an illustrative example\footnote{Ket's compiler can automatically decompose multi-controlled gates.}, the optimizations presented in this work allow a high-level implementation to match the performance of the compiler's own specialized, hand-tuned code.

The recursive function \texttt{v\_chain} in Figure~\ref{code:vchain} is structured to allow the compiler to apply the optimizations from Theorem~\ref{theorem:apx} and Theorem~\ref{theorem:alloc}. This enables the automatic use of approximate Toffoli gate decompositions and ensures the safe management of auxiliary qubits. In contrast, the function \texttt{v\_chain\_x} implements the same logic without the \texttt{with around} instruction. Consequently, the compiler cannot apply the approximate decomposition, and the auxiliary qubits must be managed explicitly by the programmer.

\begin{figure}[htbp]
    \centering
    \begin{minipage}{.54\linewidth}
        \begin{minted}{python3}
@using_aux(a=lambda c: int(len(c) > 2))
def v_chain(c, t, a):
    if len(c) <= 2:
        ctrl(c, X)(t)
    else:
        with around(ctrl(c[:2], X), a):
            v_chain(a + c[2:], t)

    \end{minted}
    \end{minipage}
    \hfill
    \begin{minipage}{.43\linewidth}
        \begin{minted}{python3}
def v_chain_x(c, t, a):
    if len(c) <= 2:
        ctrl(c, X)(t)
    else:
        ctrl(c[:2], X)(a[0])
        v_chain_x(a[0] + c[2:],
                  t, a[1:])
        ctrl(c[:2], X)(a[0])
    \end{minted}
    \end{minipage}
    \caption{Two implementations of the V-chain multi-controlled NOT decomposition. \textsc{Left}:~A~concise version using the \texttt{with around} construct and the \texttt{@using\_aux} decorator for auxiliary qubit allocation. \textsc{Right}: An explicit implementation that manually manages the auxiliary qubits and gate calls.}
    \label{code:vchain}
\end{figure}

To facilitate the dynamic allocation of auxiliary qubits, we introduce the \texttt{@using\_aux} function decorator. This feature allows the programmer to define the scope and lifetime of auxiliary qubits precisely. The decorator takes keyword arguments where each key defines the name of an auxiliary qubits (\textit{e.g.}, \texttt{a}), and the value is a lambda function that calculates the number of qubits to allocate based on the main function's arguments (\textit{e.g.}, \texttt{c}). In the \texttt{v\_chain} example, one auxiliary qubit is allocated for the variable \texttt{a} if the number of control qubits in \texttt{c} is greater than two. This allocation is re-evaluated at each step of the recursion. As the auxiliary qubits are managed automatically, it is injected as an argument into the function's scope but does not need to be passed in the recursive call, simplifying the function's signature. In contrast, \texttt{v\_chain\_x} requires the auxiliary qubits to be passed and managed manually at every step.

The structure of the \texttt{v\_chain} function's recursive step matches the conditions of our theorems. The call \texttt{with around(ctrl(c[:2], X), a)} defines the permutation gate $P$ as a Toffoli gate acting on an auxiliary qubit. The inner recursive call, \texttt{v\_chain(a + c[2:], t)}, uses the auxiliary qubit \texttt{a} only as a control. This satisfies the conditions of Theorem~\ref{theorem:apx}, allowing the compiler to substitute the Toffoli gate $P$ with a more efficient approximate version. It also satisfies Theorem~\ref{theorem:alloc}, guaranteeing that the auxiliary qubit is correctly uncomputed and can be safely freed.

The implementation using the \texttt{with around} instruction is not only easier to write and maintain, but it also generates a more efficient circuit, as shown in Figure~\ref{fig:vchain_circuit}. By enabling the compiler to use an approximate decomposition for the Toffoli gates (which uses 3 CNOTs instead of 6), the overall resource cost is significantly reduced. For the 6-control case depicted, this optimization reduces the CNOT count from 54 to 30. As the number of controls increases, this can represent up to a 50\% CNOT reduction.

\begin{figure}[htbp]
    \centering
    \includegraphics[width=\linewidth]{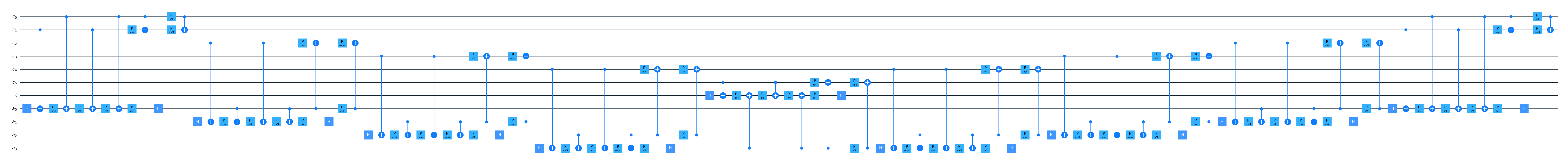}
    \vspace{5mm}
    \includegraphics[width=\linewidth]{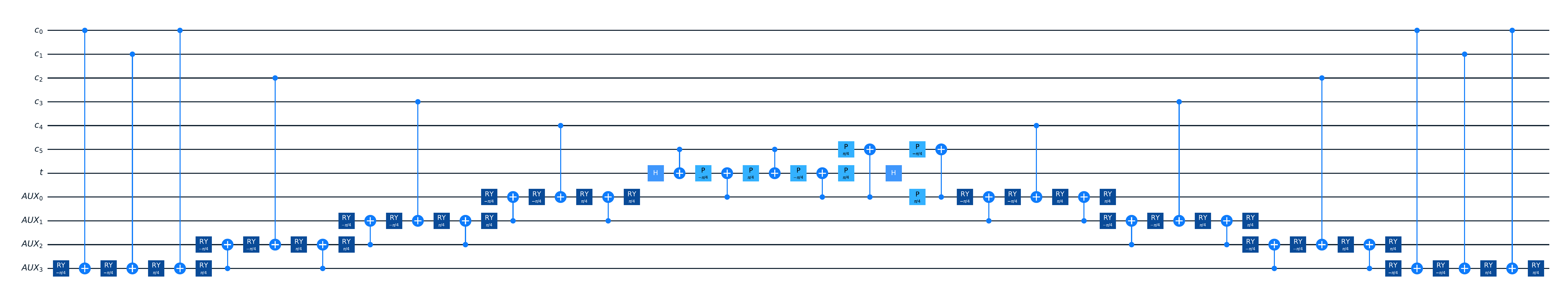}
    \caption{Resulting quantum circuits for a 6-control NOT gate using the V-chain decomposition. \textsc{Top}: The circuit generated from the explicit \texttt{v\_chain\_x} implementation. \textsc{Bottom}: The more efficient circuit generated from the \texttt{v\_chain} implementation, which leverages compiler optimizations enabled by \texttt{with around} to use approximate Toffoli.}
    \label{fig:vchain_circuit}
\end{figure}

\section{Final Remarks}\label{sec:conclusion}

In this paper, we leveraged the high-level \texttt{with around} instruction from the Ket platform, which implements the common $U^\dagger V U$ pattern, to enable compiler optimizations. We presented two primary contributions. First, we developed a theorem demonstrating how this instruction allows a compiler to automatically substitute gates with more efficient, approximate decompositions, leading to circuits with fewer CNOT gates. Second, we introduced a construction that uses the same instruction to guarantee the safe use and correct uncomputation of auxiliary qubits, simplifying dynamic memory management for quantum programming.

The implications of this work suggest a paradigm shift in how we approach high-performance quantum programming. This trend parallels the evolution of classical computing, where compilers for high-level languages eventually surpassed the performance of most hand-written assembly code. We argue that a similar trajectory is emerging in quantum computing; high-level instructions, far from being mere conveniences, can provide compilers with crucial semantic information, enabling them to explore a broader optimization space than is feasible through manual, low-level tuning. We demonstrated this concretely by implementing a V-chain decomposition for a multi-controlled-NOT gate, which, through our optimizations, can achieve up to a 50\% reduction in CNOTs. While the underlying decomposition methods are known, our work shows how they can be automatically applied by the compiler, thanks to the high-level abstraction.

The strategy proposed for verifying the safe use of auxiliary qubits are intentionally pessimistic to guarantee correctness, which opens opportunities for future research. A key direction is to develop more sophisticated static analysis techniques to identify a broader range of valid constructions. Furthermore, determining if an arbitrary unitary is a permutation or diagonal gate is not always trivial. For example, the QFT-based adder~\cite{ruiz-perezQuantumArithmeticQuantum2017} is a permutation gate, but this property is not easily identified using the definition provided in Section~\ref{subsec:opt:decomp}.

Further work could also extend these safety guarantees to the use of dirty auxiliary qubits, which are initialized in an unknown state. Finally, the dynamic allocation of auxiliary qubits creates new optimization opportunities in the circuit mapping stage, where the compiler has greater freedom to assign roles to physical qubits with favorable connectivity or lower error rates. Investigating these mapping strategies is a promising direction for future work.

\begin{credits}
    \subsubsection{\ackname} ECRR acknowledges the Coordenação de Aperfeiçoamento de Pes\-so\-al de Ní\-vel Superior - CAPES, Finance Code 001; EID, JM, and ECRR acknowledges the Conselho Nacional de De\-sen\-vol\-vi\-men\-to Ci\-en\-tí\-fi\-co e Tec\-no\-ló\-gi\-co - CNPq through grant number \mbox{409673/2022-6}; JM, EID, and ECRR acknowledges the Fundação de Amparo à Pesquisa e I\-no\-va\-ção do Estado de Santa Catarina - FAPESC through Project FAPESC TR nº 2024TR002672.
\end{credits}

\bibliographystyle{splncs04}
\bibliography{main}

\end{document}